\documentclass[conference]{IEEEtran}
\IEEEoverridecommandlockouts
\usepackage{graphicx}
\usepackage{array}
\usepackage{amsmath}
\usepackage{float}
\usepackage{multirow}
\usepackage{array}
\usepackage{cancel}

\usepackage{tabu}
\usepackage{amsthm}
\usepackage{epsfig}
\usepackage{epstopdf}
\usepackage{epsf}
\usepackage{color}
\usepackage[english]{babel}

\theoremstyle{definition}
\newtheorem{theorem}{Theorem}

\newtheorem{definition}{Definition}

\usepackage{amssymb}
\usepackage{cite}
\usepackage{amsmath,amssymb,amsfonts}
\usepackage{graphicx}
\usepackage{textcomp}
\usepackage{xcolor}
\usepackage{algorithm}
\usepackage[noend]{algpseudocode}
\usepackage[T1]{fontenc}
\usepackage{graphicx,lipsum,afterpage,subcaption}
\usepackage{enumitem}
\usepackage{lipsum}
\usepackage{mathtools}
\usepackage{cuted}

\usepackage[top=0.73in, bottom=1.1in, left=0.63in, right=0.63in]{geometry}

\usepackage{comment}
\usepackage{optidef}
\newcommand{\myheading}[1]{\noindent\textbf{#1}\hspace{0.5em}}

\begin{document}

\title{Universal Rate-Distortion-Classification Representations for Lossy Compression}
	

\author{\IEEEauthorblockN{$\textnormal{Nam Nguyen}^{1}$, $\textnormal{Thuan Nguyen}^{2}$, $\textnormal{Thinh Nguyen}^{1}$, and $\textnormal{Bella Bose}^{1}$ \thanks{This work was supported by the National Science Foundation Grant CCF-2417898.}}
\IEEEauthorblockA{$^{1}\textnormal{School of Electrical and Computer Engineering, Oregon State University, Corvallis, OR, 97331}$\\
$^{2}\textnormal{Department of Engineering, Engineering Technology, and Surveying, East Tennessee State University, Johnson, TN, 37604}$\\}
Emails: nguynam4@oregonstate.edu,  nguyent11@etsu.edu, thinhq@eecs.oregonstate.edu, bella.bose@oregonstate.edu}
	
\maketitle

\begin{abstract}
In lossy compression, Wang et al.~\cite{Wang2024} recently introduced the rate-distortion-perception-classification function, which supports multi-task learning by jointly optimizing perceptual quality, classification accuracy, and reconstruction fidelity. Building on the concept of a universal encoder introduced in~\cite{UniversalRDPs}, we investigate the universal representations that enable a broad range of distortion-classification tradeoffs through a single shared encoder coupled with multiple task-specific decoders. We establish, through both theoretical analysis and numerical experiment, that for Gaussian source under mean-squared error (MSE) distortion, the distortion-classification tradeoff region can be achieved using a single universal encoder. For general sources, we characterize the achievable region and identify conditions under which a universal encoder can produce a small distortion penalty. The experimental result on the MNIST dataset further supports our theoretical findings. We show that universal encoders can obtain distortion performance comparable to task-specific encoders. These results demonstrate the practicality and effectiveness of the proposed universal framework in multi-task compression scenarios.
\end{abstract}

\renewcommand\IEEEkeywordsname{Keywords}

\section{Introduction}
Rate-distortion theory has long served as the foundation for lossy compression, characterizing the minimum distortion achievable at a given bit rate~\cite{cover1999elements}. Conventional systems are typically evaluated using full-reference distortion metrics such as MSE, PSNR, SSIM, and MS-SSIM~\cite{wang2004image, wang2003multiscale}. However, recent research has demonstrated that minimizing distortion alone is insufficient to yield perceptually convincing reconstructions. This limitation is particularly evident in deep learning-based image compression, where empirical evidence suggests that improvements in perceptual quality often come at the expense of increased distortion~\cite{agustsson2019generative,blau2018perception}. 

To address this limitation, Blau and Michaeli~\cite{blau2019rethinking} introduced the rate-distortion-perception (RDP) framework, which incorporates perceptual quality, measured via distributional divergence, as an another metric. The RDP formulation shows a fundamental tradeoff among rate, distortion fidelity, and perceptual realism. Practical implementations, particularly those using GANs~\cite{goodfellow2014generative}, have demonstrated high perceptual quality at low bit rates~\cite{arjovsky2017wasserstein, tschannen2018deep, nowozin2016f, larsen2016autoencoding}. Common no-reference perceptual metrics include FID~\cite{heusel2017gans}, NIQE~ \cite{mittal2011blind}, PIQE~\cite{mittal2012making}, and BRISQUE~\cite{venkatanath2015blind}.

Extending this approach, the classification-distortion-perception (CDP) framework was introduced in \cite{CDP}, incorporating classification accuracy as an additional metric. This work showed that the tradeoffs among distortion, perceptual quality, and classification accuracy are fundamentally irreconcilable; that is, improving one metric generally degrades the others.

Recent work has explored integrating classification into lossy compression frameworks to jointly optimize multiple downstream objectives. The rate-distortion-classification (RDC) framework, proposed by Zhang et al.~\cite{Zhang2023}, formalized the joint optimization of rate, distortion, and classification accuracy, and established desirable properties such as monotonicity and convexity under multi-distribution source models. Extending this, Wang et al.~\cite{Wang2024} introduced the rate-distortion-perception-classification (RDPC) function, which characterizes the tradeoffs among rate, distortion, perception, and classification. Their results show that improving classification typically increases distortion or degrades perceptual quality.

These tradeoffs raise a fundamental question: are they inherently tied to the encoder's representation, or can a single encoder flexibly support different task objectives by using different decoders?  To address this, the universal RDP framework was introduced in~\cite{UniversalRDPs}, in which a fixed encoder is paired with multiple decoders to achieve various operating points in the distortion-perception space, without requiring encoder retraining.

Motivated by the idea in~\cite{UniversalRDPs}, we investigate the universal representations that enable a broad range of distortion-classification tradeoffs through a single shared encoder coupled with multiple task-specific decoders. We establish, through both theoretical analysis and numerical experiment, that for Gaussian source under mean-squared error (MSE) distortion, the distortion-classification tradeoff region can be achieved using a single universal encoder. For general sources, we characterize the achievable region and identify conditions under which a universal encoder can produce a small distortion penalty. The experimental result on the MNIST dataset further supports our theoretical findings. We show that universal encoders can obtain distortion performance comparable to task-specific encoders. These results demonstrate the practicality and effectiveness of the proposed universal framework in multi-task compression scenarios.

\section{Rate-Distortion-Classification Representations}
Consider a source generating observable data \( X \sim p_X \) with multiple latent target labels \( S_1, \dots, S_K \sim p_{S_1, \dots, S_K} \). These labels are correlated with \( X \) under the joint distribution \( p_{X, S_1, \dots, S_K} \). While \( S_1, \dots, S_K \) are not directly observed, they are often inferable from \( X \). For example, if \( X \) is an image, classification tasks can be object recognition or scene understanding.
\begin{figure}[h]
\centering
\vspace{-0.2cm}
\includegraphics[width=0.45\textwidth]{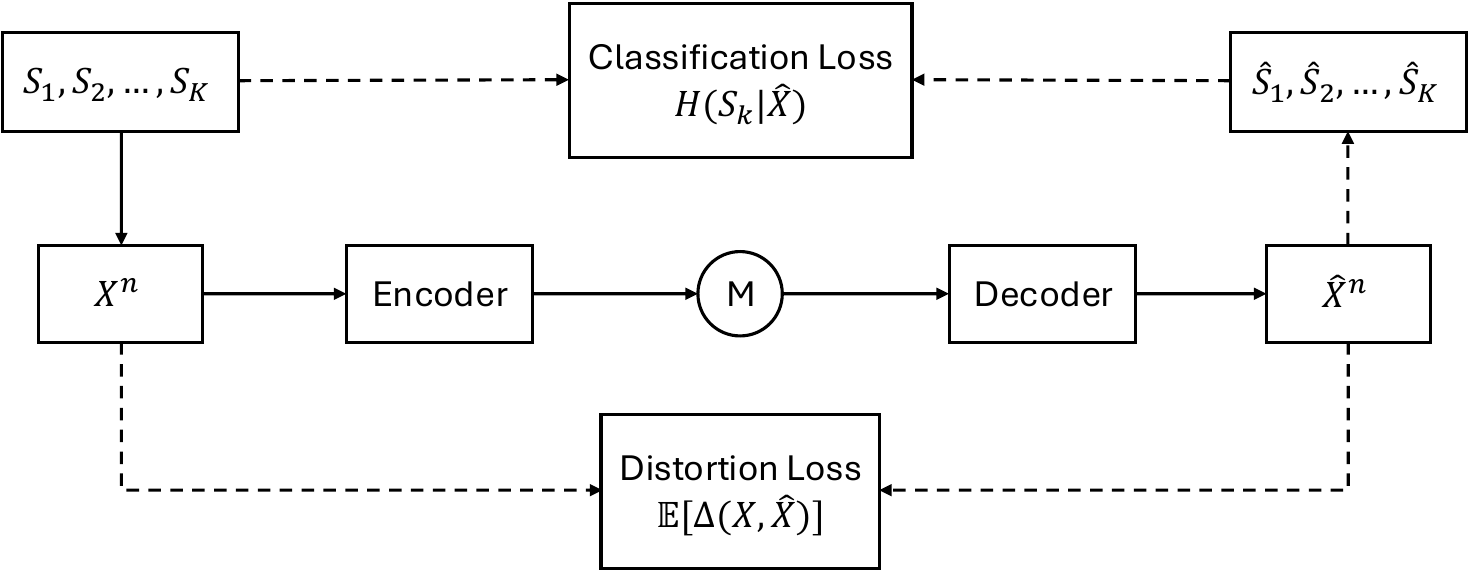}
\caption{Task-oriented lossy compression framework.}
\label{Lossy_Compression_Framework}
\end{figure}

As shown in Fig.~\ref{Lossy_Compression_Framework}, the lossy compression system consists of an encoder and a decoder. For a sequence \( X^n \sim p_X^n \), the encoder \( f: \mathcal{X}^n \to \{1,\dots,2^{nR}\} \) compresses it to a message \( M \) at rate \( R \), which is then decoded by \( g: \{1,\dots,2^{nR}\} \to \hat{\mathcal{X}}^n \) to produce \( \hat{X}^n \). The goal is to retain task-relevant information while compressing efficiently.

\myheading{Distortion constraint.}
The reconstructed signal \( \hat{X} \) must satisfy:
\begin{equation}
\mathbb{E}[\Delta(X, \hat{X})] \leq D,
\end{equation}
where \( \Delta \) is a distortion metric, e.g., Hamming or MSE.

\myheading{Classification constraint.}
We impose the following classification constraint:
\begin{equation}
    H(S_k|\hat{X}) \leq C_k, \qquad k \in [K],
    \label{ClassificationConstraint}
\end{equation}
for some $C_k > 0$. This ensures that the uncertainty of the classification variable $S_k$ given the reconstructed source $\hat{X}$ does not exceed $C_k$.

\myheading{Information RDC function.}
To jointly capture distortion and classification constraints, the information rate-distortion-classification function is defined as follows:
\begin{definition}
\cite{Wang2024}
For a source \( X \sim p_X \) and a single associated classification variable \( S \), the \emph{information rate-distortion-classification function} is defined as:
\begin{mini!}|s|[2] 
{p_{\hat{X}|X}} 
{I(X; \hat{X})} 
{\label{RDC}} 
{R(D, C) =} 
\addConstraint{\mathbb{E}[\Delta(X, \hat{X})]}{\leq D} 
\addConstraint{H(S | \hat{X})}{\leq C.}
\end{mini!}
\end{definition}

\begin{figure*}[htbp]
\begin{align}
\label{Theorem2}
D(C, R) = 
\begin{cases} 
    \sigma_X^2 e^{-2R}, \hspace{0.65cm} C  > \frac{1}{2}\log\left( 1-\frac{\theta_1^2(\sigma_X^2 - \sigma_X^2 e^{-2R})}{\sigma_S^2\sigma_X^4} \right) + h(S) \\
    \sigma_X^2 - \frac{\sigma_S^2 \sigma_X^4}{\theta_1^2} \left(1 - e^{-2h(S) + 2C}\right), \\ 
    \hspace{2cm} \frac{1}{2} \log\left(1 - \frac{\theta_1^2}{\sigma_S^2 \sigma_X^2} \right) + h(S) \leq C  \leq  \frac{1}{2}\log\left( 1-\frac{\theta_1^2(\sigma_X^2 - \sigma_X^2 e^{-2R})}{\sigma_S^2\sigma_X^4} \right) + h(S)\\
    0, \hspace{1.7cm}  C > h(S) \text{ and } R > h(X).
\end{cases}
\end{align}
\hrulefill
\end{figure*}

\section{Gaussian Source} \label{sec-RDC-Gaussian}
This section analyzes the RDC tradeoff for a scalar Gaussian source, deriving closed-form expressions that highlight the relationship between compression, distortion, and classification.

For a scalar Gaussian source, the closed-form expression of \( R(D,C) \) is given in the following theorem by Wang et al.~\cite{Wang2024}.
\begin{theorem}\label{TheoremRDCGS}
\cite{Wang2024} Let \( X\sim \mathcal{N}(\mu_X,\sigma_X^2) \) be a Gaussian source and \( S\sim \mathcal{N}(\mu_S,\sigma_S^2) \) be an associated classification variable, with a covariance of \( \text{Cov}(X,S) = \theta_1 \). The problem (\ref{RDC}) is feasible if $C \geq \frac{1}{2} \log\left(1 - \frac{\theta_1^2}{\sigma_S^2 \sigma_X^2}\right) + h(S)$. For the MSE distortion (i.e., \( \mathbb{E}[\Delta(X, \hat{X})] = \mathbb{E}[(X-\hat{X})^2] \)), the \emph{information rate-distortion-classification function} is achieved by a jointly Gaussian estimator \( \hat{X} \) and given by
  \begin{align*}
    R(D,C) \! = \!
    \begin{cases}
        \frac{1}{2} \log \frac{\sigma_X^2}{D}, \\
        \hspace{1.4cm} D \leq \sigma_X^2 \left( 1 - \frac{1}{\rho^2} \left( 1 - e^{-2h(S) + 2C} \right) \right) \\
        -\frac{1}{2} \log \left(1 - \frac{1}{\rho^2} \left( 1 - e^{-2h(S) + 2C} \right) \right), \\
        \hspace{1.4cm} D > \sigma_X^2 \left( 1 - \frac{1}{\rho^2} \left( 1 - e^{-2h(S) + 2C} \right) \right) \\
        0, \hspace{1.1cm} C > h(S) \text{ and } D > \sigma_X^2.
    \end{cases}
  \end{align*}
where \( \rho = \frac{\theta_1}{\sigma_S \sigma_X} \) represents the correlation coefficient between \( X \) and \( S \), while \( h(\cdot) \) denotes the differential entropy.
\end{theorem}
 
\begin{figure}[htbp]
\centering
\includegraphics[width=0.4\textwidth]{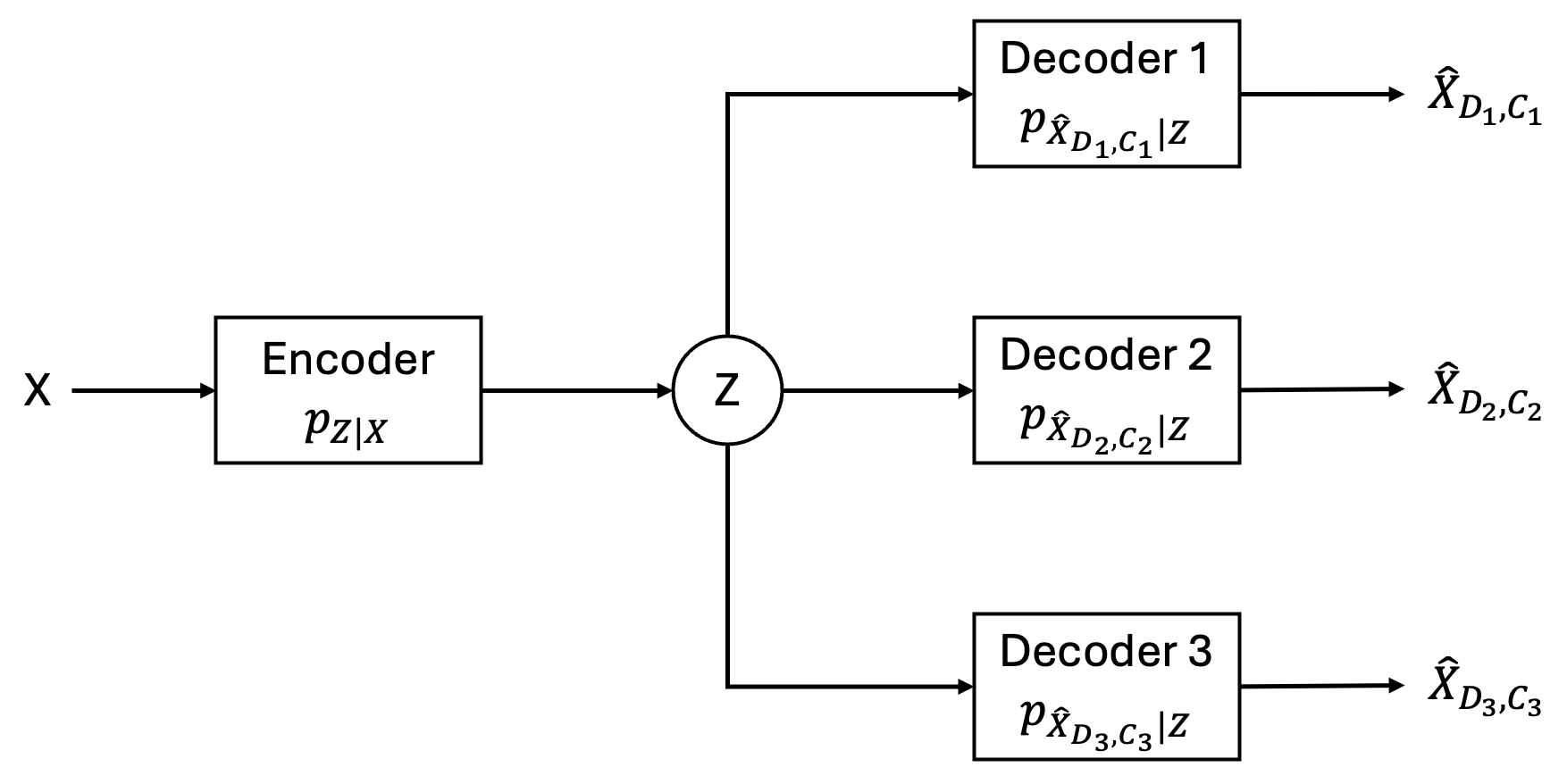}
\caption{The universal representation framework.}
\label{fig:Universal}
\end{figure}

It is noted that \( R(D, C) \) is both convex and monotonically non-increasing in \( D \) and \( C \), as established in~\cite{Wang2024}. Furthermore, we characterize the minimum achievable distortion as a function of \( C \) and \( R \) by the following definition.
\begin{definition}
For a source \( X \sim p_X \) and classification variable \( S \), the \emph{information distortion-classification-rate (DCR) function} is defined as:
\begin{mini!}|s|[2] 
{p_{\hat{X}|X}} 
{\mathbb{E}[(X-\hat{X})^2]} 
{\label{DCR_Definition}} 
{D(C,R) =} 
\addConstraint{I(X;\hat{X})}{\leq R}{} 
\addConstraint{H(S | \hat{X})}{\leq C.}{} 
\end{mini!}
\end{definition}

Our first contribution is the derivation of a closed-form expression for \( D(C, R) \) in the Gaussian source setting, as formally stated in Theorem~\ref{TheoremDCR_GS}.
\begin{theorem}\label{TheoremDCR_GS}
Consider a Gaussian source \( X\sim \mathcal{N}(\mu_X,\sigma_X^2) \) and an associated classification variable \( S\sim \mathcal{N}(\mu_S,\sigma_S^2) \) with covariance \( \text{Cov}(X,S) = \theta_1 \). The problem (\ref{DCR_Definition}) is feasible if the classification loss satisfies $C \geq \frac{1}{2} \log\left(1 - \frac{\theta_1^2}{\sigma_S^2 \sigma_X^2}\right) + h(S)$. Under the MSE distortion, the \emph{information distortion-classification-rate function} is achieved by a jointly Gaussian estimator \( \hat{X} \) and is given by (\ref{Theorem2}).
\end{theorem}
\begin{proof} 
The proof is provided in Appendix \ref{Appendix_Proof_DCR_GS}.
\end{proof}

Following the results in~\cite{Wang2024}, it can also be shown that the distortion-classification-rate function \( D(C, R) \) is convex and monotonically non-increasing in both \( C \) and \( R \). For any fixed \( R \), as \( C \) increases from \( \frac{1}{2} \log\left(1 - \frac{\theta_1^2}{\sigma_S^2 \sigma_X^2}\right) + h(S) \) to \( \frac{1}{2}\log\left(1 - \frac{\theta_1^2 (\sigma_X^2 - \sigma_X^2 e^{-2 R})}{\sigma_S^2\sigma_X^4} \right) +h(S) \), the distortion \( D(C,R) \) decreases monotonically from $\sigma_X^2 - \frac{\sigma_S^2 \sigma_X^4}{\theta_1^2} \left(1 - e^{-2h(S) + 2C}\right)$ to the optimal value \( \sigma^2_X e^{-2R} \). Increasing \( C \) beyond this point does not yield further improvements in distortion.

\section{Universal Representations}
Designing separate encoders for each distortion-classification constraint is often not desirable. This motivates the use of universal representations, where a single encoder supports multiple decoding constraints, each for a distinct task, as shown in Figure~\ref{fig:Universal}. This section introduces the universal RDC framework, quantifies the rate penalty, and presents theoretical results for both Gaussian and general sources.

\subsection{Definitions}
In the standard RDC setting, the minimum rate to satisfy a distortion-classification pair \( (D, C) \) is achieved by jointly optimizing the encoder and decoder. The proposed universal RDC framework extends this by fixing the encoder and allowing the decoder to adapt, thereby supporting all constraint pairs \( (D, C) \in \Theta \), where $\Theta$ is a given set of multiple $(D, C)$ pairs.

This raises a key question: What is the minimal additional rate required to satisfy all constraints in \( \Theta \) with a single encoder? Ideally, this rate penalty is small, indicating near-optimal performance across tasks. We formalize this notion using the information universal rate-distortion-classification function, adapting the definition from~\cite{UniversalRDPs}, as follows.
\begin{definition}
Let \( Z \sim p_{Z|X} \) be a representation of \( X \). Define \( \mathcal{P}_{Z|X}(\Theta) \) as the set of encoders such that for every \( (D, C) \in \Theta \), there exists a decoder \( p_{\hat{X}_{D,C}|Z} \) satisfying $\mathbb{E}[\Delta(X, \hat{X}_{D,C})] \leq D, \quad H(S | \hat{X}_{D,C}) \leq C,$
with \( X \leftrightarrow Z \leftrightarrow \hat{X}_{D,C} \). The universal RDC function is defined as
\begin{equation}
\label{eqn:R_phi}
R(\Theta) = \inf_{p_{Z|X} \in \mathcal{P}_{Z|X}(\Theta)} I(X; Z).
\end{equation}
\end{definition}

A representation \( Z \) is \(\Theta\)-universal if it satisfies all constraints in \( \Theta \) via appropriate decoders.

Similarly, we adapt the definition from~\cite{UniversalRDPs} to quantify the rate penalty under the classification constraint as follows:
\begin{definition}
The \emph{rate penalty} is
\begin{equation}
A(\Theta) = R(\Theta) - \sup_{(D, C) \in \Theta} R(D, C),
\end{equation}
which quantifies the extra rate required for universality.
\end{definition}

Let \( \Omega(R) = \{ (D, C) : R(D, C) \leq R \} \) denote the set of achievable distortion-classification pairs at rate \( R \), and define
\begin{equation*}
\Omega(p_{Z|X}) =
\left\{ (D, C) : \exists\, p_{\hat{X}_{D,C}|Z} \text{ s.t. } 
\begin{aligned}
    &\mathbb{E}[\Delta(X, \hat{X}_{D,C})] \leq D, \\
    & H(S | \hat{X}_{D,C}) \leq C
\end{aligned}
\right\}.
\end{equation*}
If \( I(X; Z) = R \) and \( \Omega(p_{Z|X}) = \Omega(R) \), then \( Z \) achieves the maximal distortion-classification region at rate \( R \).

\subsection{Main Results}
\begin{theorem}\label{Theorem_gaussian_universality}
Let \( X \sim \mathcal{N}(\mu_X, \sigma_X^2) \) be a scalar Gaussian source and \( S \sim \mathcal{N}(\mu_S, \sigma_S^2) \) a classification variable with \( \mathrm{Cov}(X, S) = \theta_1 \). Assume MSE distortion and classification loss measured by \( H(S | \hat{X}) \). Then, for any non-empty set \( \Theta \) of distortion-classification pairs \( (D, C) \), the rate penalty is zero:
\begin{equation}
    A(\Theta) = 0.
\end{equation}

Moreover, any jointly Gaussian representation \( Z \) of \( X \) satisfying
\begin{equation}
    I(X; Z) = \sup_{(D, C) \in \Theta} R(D, C),
\end{equation}
achieves the maximal distortion-classification region:
\begin{equation}\label{eqn:sup_rate}
    \Theta \subseteq \Omega(p_{Z|X}) = \Omega(I(X; Z)).
\end{equation}
\end{theorem}
\begin{proof} 
The detailed proof is presented in Appendix~\ref{Appendix_Proof_GS_Universality}.
\end{proof}

In addition, we consider a general source \( X \sim p_X \) and characterize the distortion-classification region induced by an arbitrary representation \( Z \) under MSE distortion.

\begin{theorem}\label{Theorem_general_universality}
Let \( X \sim p_X \) be a general source and \( S \) a classification variable with \( \mathrm{Cov}(X, S) = \theta_1 \). Assume distortion is measured by MSE and classification loss by \( H(S | \hat{X}) \). Let \( Z \) be any representation of \( X \), and define \( \tilde{X} = \mathbb{E}[X | Z] \) as the minimum mean square estimator. Then the closure of the achievable region, \( \mathrm{cl}(\Omega(p_{Z | X})) \), satisfies
\begin{equation*}
\begin{split}
    \Omega(p_{Z|X}) \! &\subseteq \! \left\{ (D,C) : D \geq \mathbb{E}{\|X-\tilde{X}\|^2} \! + \! \begin{aligned}
    &\inf_{p_{\hat{X}}}   W^2_2(p_{\tilde{X}},p_{\hat{X}}) \\
    &\text{s.t. } H(S|\hat{X}) \leq C
\end{aligned} \right\} \\
    \! & \subseteq \! \mbox{cl}(\Omega(p_{Z|X})),
\end{split}
\end{equation*}
where the squared 2-Wasserstein distance is $W_2^2(p_X, p_{\hat{X}}) = \inf_{p_{X,\hat{X}}} \mathbb{E}[\|X - \hat{X}\|^2]$ with the infimum taken over all joint distributions with marginals \( p_X \) and \( p_{\hat{X}} \).

Moreover, \( \mathrm{cl}(\Omega(p_{Z|X})) \) contains the extreme points:
\begin{align*}
(D^{(a)}, C^{(a)}) &= \left( \mathbb{E}[\|X-\tilde{X}\|^2], \sum_{s}\sum_{\tilde{x}} p_{\tilde{X}} p_{S|\tilde{X}} \log\frac{1}{p_{S|\tilde{X}}} \right), \\
(D^{(b)}, C^{(b)}) &= \left( \mathbb{E}[\|X - \tilde{X}\|^2] + W_2^2(p_{\tilde{X}}, p_{\hat{X}^{C_{\text{min}}}}), C_{\min} \right),
\end{align*}
where \begin{mini!}|s|[2]
    {p_{\hat{X}}} 
    {H(S | \hat{X})} 
    {\label{C_min}} 
    {p_{\hat{X}^{C_{\text{min}}}} = \arg} 
    \addConstraint{\mathbb{E}[\|X - \hat{X}\|^2]}{\leq D.}
\end{mini!}
The minimum classification loss is: $C_{\min} = \sum_{s}\sum_{\hat{x}^{C_{\text{min}}}} p_{\hat{X}^{C_{\min}}} p_{S|\hat{X}^{C_{\min}}} \log\frac{1}{p_{S|\hat{X}^{C_{\min}}}}$.
\end{theorem}
\begin{proof} 
A complete proof is provided in Appendix \ref{Appendix_Proof_General_Universality}.
\end{proof}

\begin{figure}[!htbp]
\centering
\includegraphics[width=0.42\textwidth]{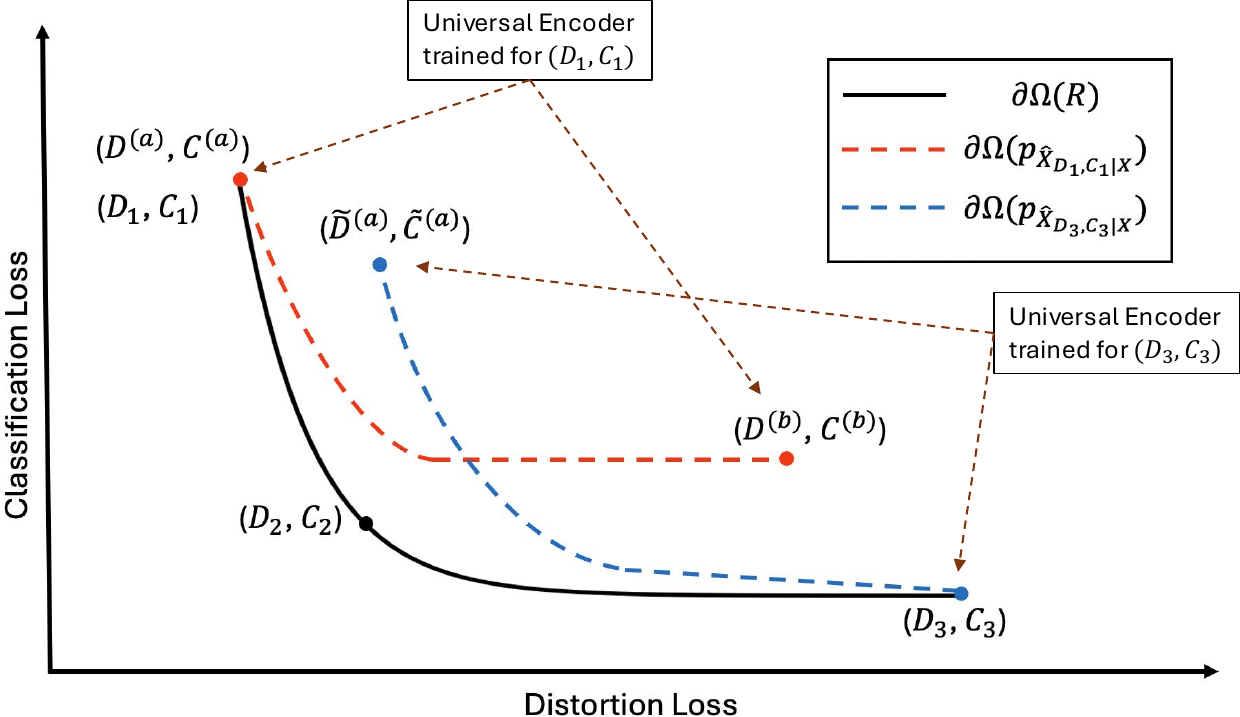}
\caption{Universality for a general source. Shown are the boundaries of achievable distortion-classification regions for three representations: the minimal distortion point \( (D_1, C_1) \), where \( R(D_1, C_1) = R(D_1, \infty) \); the midpoint \( (D_2, C_2) \); and the minimal classification loss point \( (D_3, C_3) \).}
\label{Fig:Universal_GernalSource_CDR}
\end{figure}

To further analyze the structure of the achievable region, consider a point \( (D, C) \) on the distortion-classification trade-off curve at a fixed rate \( R \), with an associated optimal reconstruction \( Z = \hat{X}_{D,C} \) satisfying \( I(X; \hat{X}_{D,C}) = R \), \( \mathbb{E}[\|X - \hat{X}_{D,C}\|^2] = D \), and \( H(S | \hat{X}_{D,C}) = C \). Assuming that such an optimal reconstruction exists for every point on the trade-off curve and that any further reduction in either \( D \) or \( C \) would violate the rate constraint, it follows that the point \( (D, C) \) lies on the boundary of the closure \( \mathrm{cl}(\Omega(p_{\hat{X}_{D,C} \mid X})) \).

By Theorem~\ref{Theorem_general_universality}, the closure of the achievable region includes two extreme points: the upper-left point \( (D^{(a)}, C^{(a)}) \), which minimizes distortion, and the lower-right point \( (D^{(b)}, C^{(b)}) \), which minimizes classification loss. Both points are realized by the universal encoder trained at the operating point \( (D_1, C_1) \). The region is convex and contains all intermediate pairs. Figure~\ref{Fig:Universal_GernalSource_CDR} illustrates both \( \Omega(R) \) and the achievable region \( \Omega(p_{\hat{X}_{D,C} \mid X}) \) for representative points on the trade-off curve. The following theorem quantifies this structure.

\begin{theorem}\label{Theorem_Quantitative_Results}
Let \( \hat{X}_{D_1, C_1} \) denote the optimal reconstruction at point \( (D_1, C_1) \) on the conventional RDC trade-off curve, satisfying \( I(X; \hat{X}_{D_1, C_1}) = R(D_1, C_1) \). Then the upper-left extreme point of \( \Omega(p_{\hat{X}_{D_1, C_1} \mid X}) \) satisfies $(D^{(a)}, C^{(a)}) = (D_1, C_1)$. Now consider the lower-right extreme points: \( (D^{(b)}, C^{(b)}) \in \Omega(p_{\hat{X}_{D_1, C_1} \mid X}) \) and \( (D_3, C_3) \in \Omega(R) \), where \( C_3 = C_{\min} \) and \( R(D_3, C_3) = R(D_1, \infty) \). The distortion gap between these points is bounded below by:
\begin{equation}
D_3 - D^{(b)} \! \geq \! \sigma_X^2 + \sigma_{\hat{X}_{D_3, C_3}}^2  \!\! - 2 \sigma_{\hat{X}_{D_3, C_3}} \sqrt{\sigma_X^2 - D_1} - 2D_1,
\end{equation}
and the corresponding distortion ratio satisfies:
\begin{equation}
\frac{D_3}{D^{(b)}} \geq 
\frac{\sigma_X^2 + \sigma_{\hat{X}_{D_3, C_3}}^2 
- 2 \sigma_{\hat{X}_{D_3, C_3}} \sqrt{\sigma_X^2 - D_1}}{2D_1}.
\end{equation}

In the case where \( W_2^2(p_X, p_{\hat{X}_{D_3, C_3}}) = 0 \), i.e., \( \sigma_X^2 = \sigma_{\hat{X}_{D_3, C_3}}^2 \), the distortion gap becomes small under:
\begin{equation}
\label{Distortion_Gap1}
D_3 - D^{(b)} \approx 0 \;\; \text{if} \;\; D_1 \approx 0 \text{ or } D_1 \approx \sigma_X^2,
\end{equation}
\begin{equation}
\label{Distortion_Gap2}
\frac{D_3}{D^{(b)}} \approx 1 \;\; \text{if} \;\; D_1 \approx \sigma_X^2.
\end{equation}
\end{theorem}
\begin{proof}
The proof is provided in Appendix~\ref{Appendix_Proof_Quantitative_Results}.
\end{proof}

\section{Experimental Results}
\subsection{Gaussian Zero-Rate Penalty}
\begin{figure}[h] 
\centering
\includegraphics[width=0.42\textwidth]{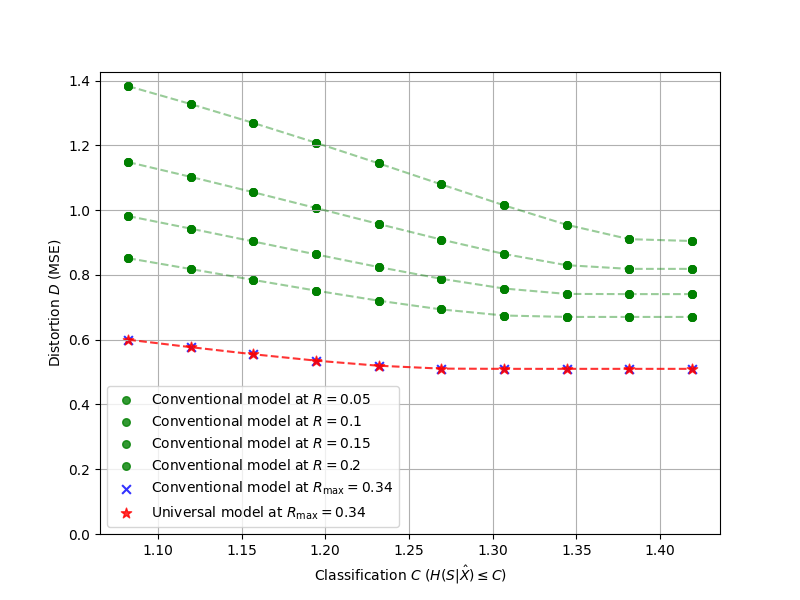}
\caption{CDR functions for a Gaussian source.}
\label{fig:DCR_Gaussian_NoPenalty}
\end{figure}

This section verifies that no rate penalty arises when replacing multiple distortion-classification specific Gaussian encoders (i.e. conventional model) with a single universal encoder (i.e., universal model). We consider a scalar Gaussian source \( X \sim \mathcal{N}(0, 1) \) and classification variable \( S \sim \mathcal{N}(0, 1) \) with correlation \( \rho = 0.7 \), yielding a maximum rate of \( R_{\max} = 0.34 \). 

The conventional model is evaluated at rates \( [0.05, 0.1, 0.15, 0.2, 0.34] \), with each rate generating a distinct set of \((C, D)\) tradeoffs via Theorem~\ref{TheoremDCR_GS}. In contrast, the universal model uses a fixed encoder at \( R_{\max} \) (by Theorem~\ref{Theorem_gaussian_universality}) and varies the decoder to explore achievable \((C, D)\) pairs.

Figure~\ref{fig:DCR_Gaussian_NoPenalty} illustrates that the universal model closely traces the boundary of the conventional model, thereby confirming that a single encoder is sufficient to achieve the entire classification-distortion tradeoff for Gaussian sources without incurring any rate penalty.

\subsection{Universal Representation for Lossy Compression}
\begin{figure}[!htbp]
\centering
\includegraphics[width=0.43\textwidth]{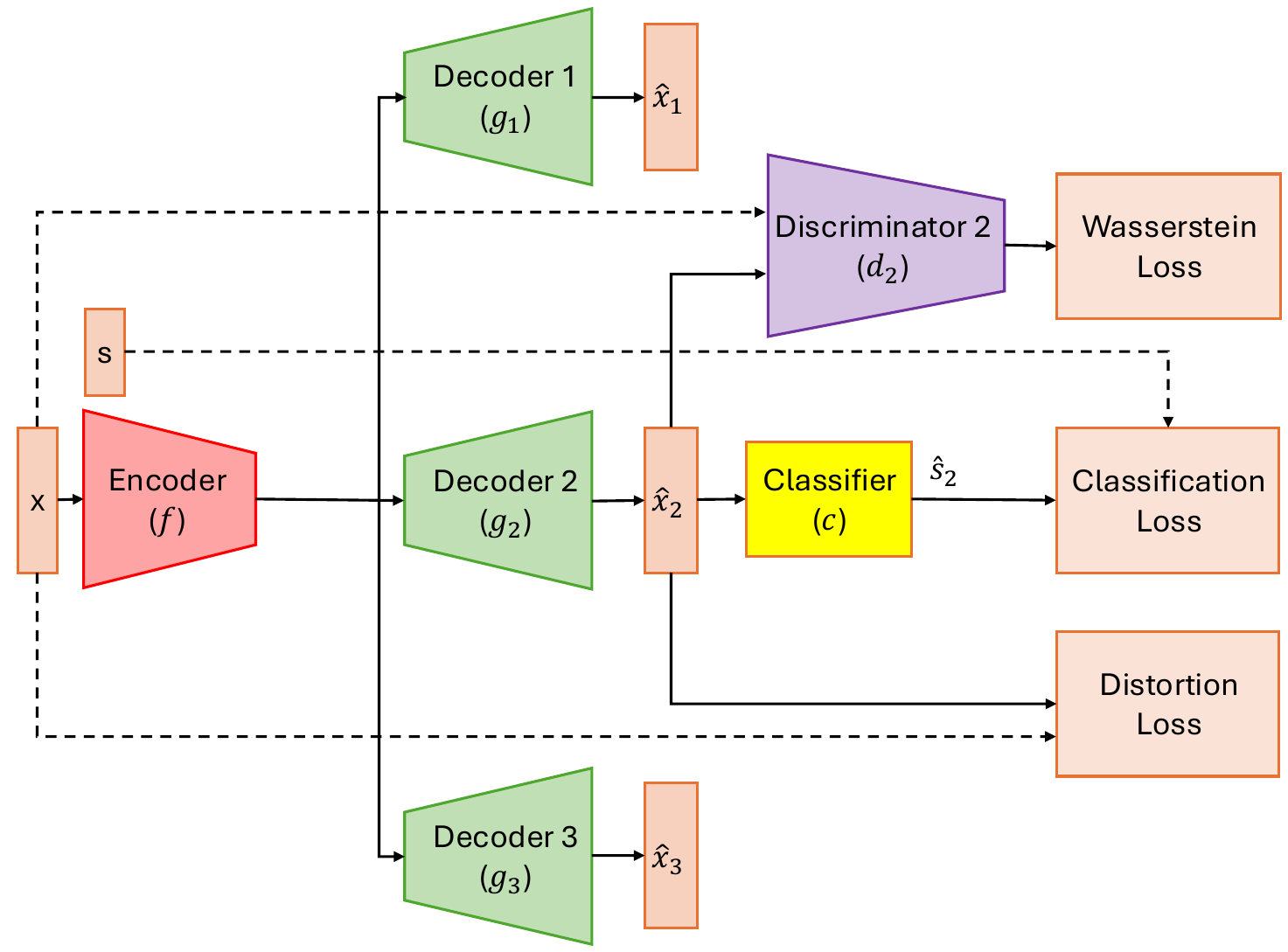}
\caption{An illustration of the universal RDC scheme.}
\label{fig:Scheme}
\end{figure}

We empirically validate our theory using a deep learning-based image compression framework on MNIST dataset, showing that the distortion gap between conventional and universal models aligns with our theoretical predictions.

\subsubsection{Training}
We use a stochastic autoencoder with a pre-trained classifier and GAN-based discriminator, consisting of an encoder \( f \), decoder \( g \), classifier \( c \), and discriminator \( d \), as shown in Figure~\ref{fig:Scheme}. In the conventional setup, \( f \), \( g \), and \( d \) are trainable. The distortion loss is measured by MSE. The output \( \hat{X} \) is passed through the classifier \( c \) to produce the predicted label distribution \( \hat{S} \), with classification loss computed via cross-entropy \( \text{CE}(S, \hat{S}) \), an upper bound on conditional entropy~\cite{boudiaf2021unifying_cross_entropy}. The compression rate is upper bounded by \( h \times \log_2(L) \), where \( h \) is the encoder output size and \( L \) the quantization level. 

To ensure that the condition \( W_2^2(p_X, p_{\hat{X}_{D_3, C_3}}) = 0 \) in Theorem~\ref{Theorem_Quantitative_Results} is satisfied, we augment the training loss with a squared 2-Wasserstein distance regularization term, rather than relying solely on MSE and cross-entropy losses. Following the approach in~\cite{UniversalRDPs} and the fact that Wasserstein-2 loss can be bounded by Wasserstein-1 up to a factor (see Appendix B in \cite{lyu2023barycentric}), we estimate the 2-Wasserstein distance using a discriminator \( d \) that takes both \( X \) and \( \hat{X} \) as inputs and employs the Wasserstein-1 loss in a GAN-based framework. The complete formulation of the loss function is provided below.
\begin{equation}
\label{eqn:experimental_loss}
\mathcal{L} = \lambda_d \, \mathbb{E}[\|X - \hat{X}\|^2] + \lambda_c \, \text{CE}(S, \hat{S}) + \lambda_p \, W_1(p_X, p_{\hat{X}}),
\end{equation}
where \( \lambda_d \), \( \lambda_c \), and \( \lambda_p \) controlling the trade-off.

To construct the universal model, the trained encoder \( f \) is frozen, and a new decoder \( g_1 \) and discriminator \( d_1 \) are trained using:
\begin{equation}
\label{eqn:experimental_loss2}
\mathcal{L}_1 = \lambda_d^1 \, \mathbb{E}[\|X - \hat{X}_1\|^2] + \lambda_c^1 \, \text{CE}(S, \hat{S}) + \lambda_p^1 \, W_1(p_X, p_{\hat{X}}),
\end{equation}
where \( \lambda_d^1 \), \( \lambda_c^1 \), and \( \lambda_p^1 \) adjust the task-specific trade-offs.

\begin{figure}[h] 
\centering
\includegraphics[width=0.40\textwidth]{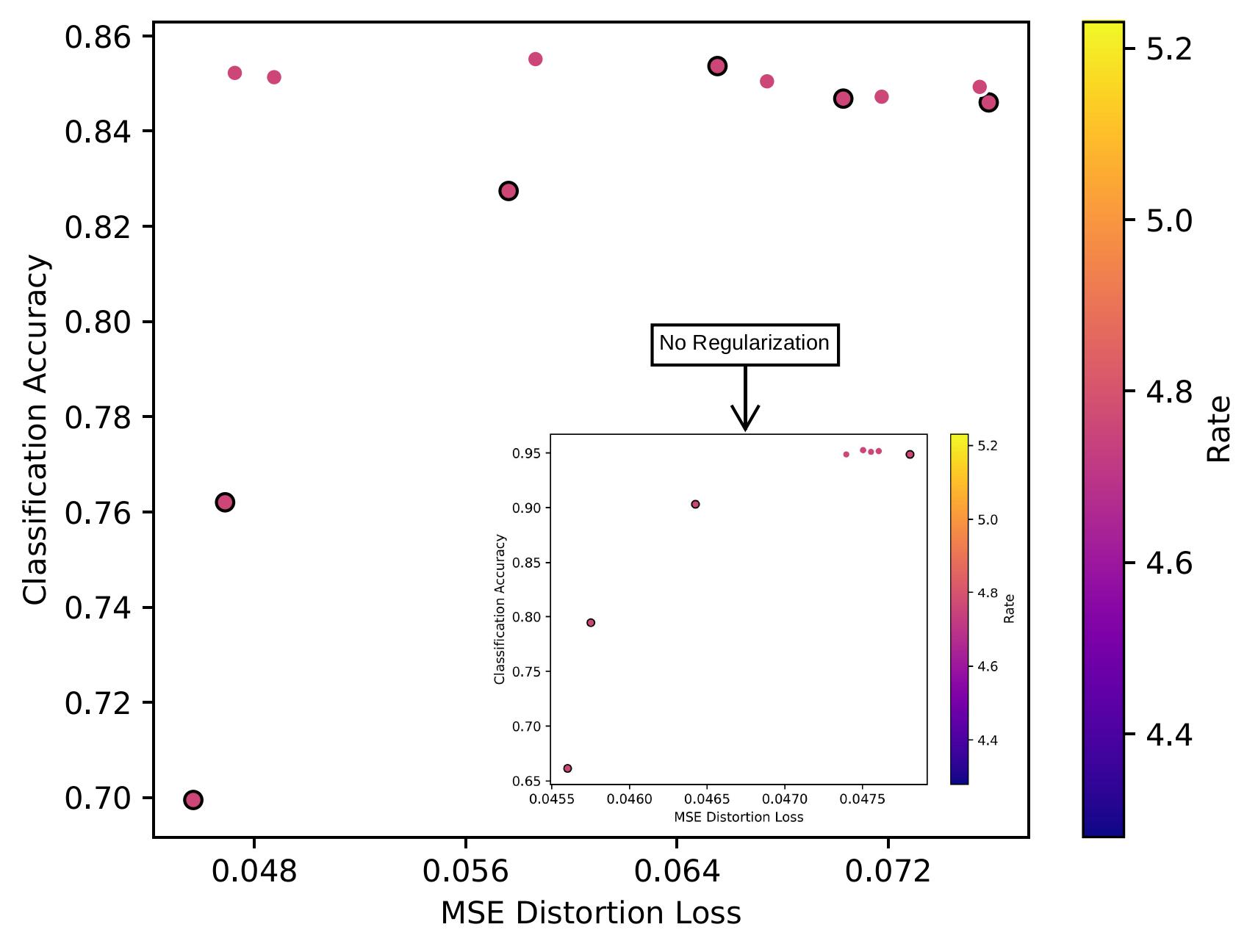}
\caption{RDC function at \( R = 4.75 \) on MNIST dataset.}
\label{fig:CDR_Comparision_Perception_MNIST}
\end{figure}
Figure~\ref{fig:CDR_Comparision_Perception_MNIST} shows the RDC tradeoff on the MNIST dataset at a fixed rate \( R = 4.75 \), obtained by varying loss coefficients. Black-outlined points represent the conventional model trained jointly for specific classification-distortion objectives. Other points correspond to the universal model with decoders trained on a fixed encoder optimized for low classification loss \( C \). As expected, in conventional models under a fixed rate constraint, reducing the cross-entropy loss (equivalently, improving classification accuracy) results in increased distortion. This observation highlights the inherent tradeoff between classification performance and reconstruction fidelity. In addition, despite using a fixed encoder, the universal model achieves distortion levels comparable to the conventional model, confirming that an encoder trained for low \( C \) can still support diverse tradeoffs through decoder retraining. These observations support the validity of Theorem~\ref{Theorem_Quantitative_Results}.

However, a noticeable classification gap remains: universal decoders cannot recover low classification (\( C \)) performance if the encoder is trained only for high-distortion objectives. This highlights the decoder's limited generative capacity when the encoder fails to preserve classification-task information.

\section{Conclusion}
We proposed a universal RDC framework that enables a single encoder to support multiple task objectives through specialized decoders, removing the need for separate encoders per distortion-classification tradeoff. For the Gaussian source with MSE distortion, we proved that the full RDC region is achievable with zero rate penalty using a fixed encoder. For general source, we characterized the achievable region using MMSE estimation and the 2-Wasserstein distance, identifying conditions under which encoder reuse incurs negligible distortion penalty. Empirical results on the MNIST dataset support our theory, showing that universal encoders, trained with Wasserstein loss regularization, achieve distortion performance comparable to task-specific models. These findings highlight the practicality and effectiveness of universal representations for multi-task lossy compression.

\newpage
\bibliographystyle{IEEEtran}
\bibliography{main}

\begin{thebibliography}{10}
\providecommand{\url}[1]{#1}
\csname url@samestyle\endcsname
\providecommand{\newblock}{\relax}
\providecommand{\bibinfo}[2]{#2}
\providecommand{\BIBentrySTDinterwordspacing}{\spaceskip=0pt\relax}
\providecommand{\BIBentryALTinterwordstretchfactor}{4}
\providecommand{\BIBentryALTinterwordspacing}{\spaceskip=\fontdimen2\font plus
\BIBentryALTinterwordstretchfactor\fontdimen3\font minus \fontdimen4\font\relax}
\providecommand{\BIBforeignlanguage}[2]{{%
\expandafter\ifx\csname l@#1\endcsname\relax
\typeout{** WARNING: IEEEtran.bst: No hyphenation pattern has been}%
\typeout{** loaded for the language `#1'. Using the pattern for}%
\typeout{** the default language instead.}%
\else
\language=\csname l@#1\endcsname
\fi
#2}}
\providecommand{\BIBdecl}{\relax}
\BIBdecl

\bibitem{Wang2024}
Y.~Wang, Y.~Wu, S.~Ma, and Y.-J.~A. Zhang, ``Lossy compression with data, perception, and classification constraints,'' in \emph{2024 IEEE Information Theory Workshop (ITW)}, 2024, pp. 366--371.

\bibitem{UniversalRDPs}
G.~Zhang, J.~Qian, J.~Chen, and A.~Khisti, ``Universal rate-distortion-perception representations for lossy compression,'' in \emph{Advances in Neural Information Processing Systems}, M.~Ranzato, A.~Beygelzimer, Y.~Dauphin, P.~Liang, and J.~W. Vaughan, Eds., vol.~34.\hskip 1em plus 0.5em minus 0.4em\relax Curran Associates, Inc., 2021, pp. 11\,517--11\,529.

\bibitem{cover1999elements}
T.~M. Cover and J.~A. Thomas, \emph{Elements of information theory}.\hskip 1em plus 0.5em minus 0.4em\relax John Wiley \& Sons, 1999.

\bibitem{wang2004image}
Z.~Wang, A.~C. Bovik, H.~R. Sheikh, and E.~P. Simoncelli, ``Image quality assessment: from error visibility to structural similarity,'' \emph{IEEE transactions on image processing}, vol.~13, no.~4, pp. 600--612, 2004.

\bibitem{wang2003multiscale}
Z.~Wang, E.~P. Simoncelli, and A.~C. Bovik, ``Multiscale structural similarity for image quality assessment,'' in \emph{The Thrity-Seventh Asilomar Conference on Signals, Systems \& Computers, 2003}, vol.~2.\hskip 1em plus 0.5em minus 0.4em\relax IEEE, 2003, pp. 1398--1402.

\bibitem{agustsson2019generative}
E.~Agustsson, M.~Tschannen, F.~Mentzer, R.~Timofte, and L.~V. Gool, ``Generative adversarial networks for extreme learned image compression,'' in \emph{Proceedings of the IEEE International Conference on Computer Vision}, 2019, pp. 221--231.

\bibitem{blau2018perception}
Y.~Blau and T.~Michaeli, ``The perception-distortion tradeoff,'' in \emph{Proceedings of the IEEE Conference on Computer Vision and Pattern Recognition}, 2018, pp. 6228--6237.

\bibitem{blau2019rethinking}
------, ``Rethinking lossy compression: The rate-distortion-perception tradeoff,'' in \emph{International Conference on Machine Learning}, 2019, pp. 675--685.

\bibitem{goodfellow2014generative}
I.~Goodfellow, J.~Pouget-Abadie, M.~Mirza, B.~Xu, D.~Warde-Farley, S.~Ozair, A.~Courville, and Y.~Bengio, ``Generative adversarial nets,'' vol.~27, 2014, pp. 2672--2680.

\bibitem{arjovsky2017wasserstein}
M.~Arjovsky, S.~Chintala, and L.~Bottou, ``Wasserstein generative adversarial networks,'' in \emph{International Conference on Machine Learning}, 2017, pp. 214--223.

\bibitem{tschannen2018deep}
M.~Tschannen, E.~Agustsson, and M.~Lucic, ``Deep generative models for distribution-preserving lossy compression,'' in \emph{Advances in Neural Information Processing Systems}, 2018, pp. 5929--5940.

\bibitem{nowozin2016f}
S.~Nowozin, B.~Cseke, and R.~Tomioka, ``f-gan: training generative neural samplers using variational divergence minimization,'' in \emph{Advances in Neural Information Processing Systems}, 2016, pp. 271--279.

\bibitem{larsen2016autoencoding}
A.~B.~L. Larsen, S.~K. S{\o}nderby, H.~Larochelle, and O.~Winther, ``Autoencoding beyond pixels using a learned similarity metric,'' in \emph{International Conference on Machine Learning}, 2016, pp. 1558--1566.

\bibitem{heusel2017gans}
M.~Heusel, H.~Ramsauer, T.~Unterthiner, B.~Nessler, and S.~Hochreiter, ``Gans trained by a two time-scale update rule converge to a local nash equilibrium,'' in \emph{Advances in Neural Information Processing Systems}, 2017, pp. 6626--6637.

\bibitem{mittal2011blind}
A.~Mittal, A.~K. Moorthy, and A.~C. Bovik, ``Blind/referenceless image spatial quality evaluator,'' in \emph{2011 conference record of the forty fifth asilomar conference on signals, systems and computers (ASILOMAR)}.\hskip 1em plus 0.5em minus 0.4em\relax IEEE, 2011, pp. 723--727.

\bibitem{mittal2012making}
A.~Mittal, R.~Soundararajan, and A.~C. Bovik, ``Making a “completely blind” image quality analyzer,'' \emph{IEEE Signal processing letters}, vol.~20, no.~3, pp. 209--212, 2012.

\bibitem{venkatanath2015blind}
N.~Venkatanath, D.~Praneeth, M.~C. Bh, S.~S. Channappayya, and S.~S. Medasani, ``Blind image quality evaluation using perception based features,'' in \emph{2015 Twenty First National Conference on Communications (NCC)}.\hskip 1em plus 0.5em minus 0.4em\relax IEEE, 2015, pp. 1--6.

\bibitem{CDP}
D.~Liu, H.~Zhang, and Z.~Xiong, ``On the classification-distortion-perception tradeoff,'' in \emph{Advances in Neural Information Processing Systems}, H.~Wallach, H.~Larochelle, A.~Beygelzimer, F.~d\textquotesingle Alch\'{e}-Buc, E.~Fox, and R.~Garnett, Eds., vol.~32.\hskip 1em plus 0.5em minus 0.4em\relax Curran Associates, Inc., 2019.

\bibitem{Zhang2023}
\BIBentryALTinterwordspacing
Y.~Zhang, ``A rate-distortion-classification approach for lossy image compression,'' \emph{Digital Signal Processing}, vol. 141, p. 104163, Sep. 2023. [Online]. Available: \url{http://dx.doi.org/10.1016/j.dsp.2023.104163}
\BIBentrySTDinterwordspacing

\bibitem{boudiaf2021unifying_cross_entropy}
M.~Boudiaf, J.~Rony, I.~M. Ziko, E.~Granger, M.~Pedersoli, P.~Piantanida, and I.~B. Ayed, ``A unifying mutual information view of metric learning: cross-entropy vs. pairwise losses,'' \emph{arXiv preprint}, 2021, arXiv 2003.08983.

\bibitem{lyu2023barycentric}
B.~Lyu, T.~Nguyen, P.~Ishwar, M.~Scheutz, and S.~Aeron, ``Barycentric-alignment and reconstruction loss minimization for domain generalization,'' \emph{IEEE Access}, vol.~11, pp. 49\,226--49\,240, 2023.

\bibitem{IBGaussian}
G.~Chechik, A.~Globerson, N.~Tishby, and Y.~Weiss, ``Information bottleneck for gaussian variables,'' in \emph{Advances in Neural Information Processing Systems}, S.~Thrun, L.~Saul, and B.~Sch\"{o}lkopf, Eds., vol.~16.\hskip 1em plus 0.5em minus 0.4em\relax MIT Press, 2003.

\bibitem{berger1999semi}
T.~Berger and R.~Zamir, ``A semi-continuous version of the berger-yeung problem,'' \emph{IEEE Transactions on Information Theory}, vol.~45, no.~5, pp. 1520--1526, 1999.

\end{thebibliography}

\clearpage
\setcounter{page}{1}
\newpage
\appendix

\subsection{Proof of Theorem \ref{TheoremDCR_GS}}\label{Appendix_Proof_DCR_GS}
Consider the distortion-classification-rate function \( D(C, R) \) under the MSE distortion criterion as follows
\begin{mini!}|s|[2] 
{p_{\hat{X}|X}} 
{\mathbb{E}[(X-\hat{X})^2]} 
{} 
{D(C,R) =} 
\addConstraint{I(X;\hat{X})}{\leq R}{} 
\addConstraint{h(S | \hat{X})}{\leq C.}{} 
\end{mini!}
where \( (X, S) \) are jointly Gaussian random variables with covariance \( \mathrm{Cov}(X, S) = \theta_1 \). The optimal solution is attained when \( \hat{X} \) is also Gaussian and jointly distributed with \( X \) \cite{Wang2024, UniversalRDPs, IBGaussian}. Indeed, we can replace any random variable $\hat{X}$ with a Gaussian random variable $\hat{X}_G$, having same mean and variance as $X$, such that (a) $\mathbb{E}[(X-\hat{X})^2] \geq \mathbb{E}[(X-\hat{X}_G)^2]$, (b) $I(X;\hat{X}) \geq I(X;\hat{X}_G)$, and (c) $h(S | \hat{X}) \geq h(S | \hat{X}_G)$. The proof for claims (a) and (b) can be found in the proof of Theorem 1 in \cite{UniversalRDPs}, and claim (c) is inherited from the entropy power inequality in \cite{IBGaussian, berger1999semi}. Thus, the optimization reduces to a parameter search over the mean \( \mu_{\hat{X}} \), variance \( \sigma_{\hat{X}}^2 \), and covariance \( \mathrm{Cov}(X, \hat{X}) = \theta_2 \).

By applying the closed-form expressions for differential entropy and mutual information of jointly Gaussian variables~\cite{cover1999elements}, we obtain:
\begin{equation}
I(X;\hat{X}) = -\frac{1}{2}\log\left(1 - \frac{\theta_2^2}{\sigma_X^2 \sigma_{\hat{X}}^2}\right),
\label{rateExpression}
\end{equation}
And for the classification constraint:
\begin{equation*}
\begin{split}  
h(S|\hat{X}) &= h(S) - I(S; \hat{X}) \leq C, \\
I(S; \hat{X}) &\geq h(S) - C, \\
-\frac{1}{2} \log\left(1 - \frac{\theta_1^2}{\sigma_S^2 \sigma_X^4} \times \frac{\theta_2^2}{\sigma_{\hat{X}}^2}\right) &\geq h(S) - C.
\end{split}
\end{equation*}

Additionally, the mean squared error between \( X \) and \( \hat{X} \) can be expressed as \cite{UniversalRDPs}:
\begin{equation}
\mathbb{E}[(X - \hat{X})^2] = (\mu_X - \mu_{\hat{X}})^2 + \sigma_X^2 + \sigma_{\hat{X}}^2 - 2 \theta_2.
\end{equation}

Then, the $D(C,R)$ problem can be formulated as:
\begin{mini!}|s|[2] 
{\mu_{\hat{X}},\sigma_{\hat{X}},\theta_2} 
{\!\!\!\!\!(\mu_X-\mu_{\hat{X}})^2 +\sigma_X^2+\sigma_{\hat{X}}^2-2\theta_2} 
{\label{DCR}} 
{D(C,R) \! = \!\!\!} 
\addConstraint{\!\!\!\!\!\!\!\!\!\!\!\!\!\!\!\! -\frac{1}{2}\log\left(1 - \frac{\theta_2^2}{\sigma_X^2 \sigma_{\hat{X}}^2}\right) \leq R}{\label{DCR_I}} 
\addConstraint{\!\!\!\!\!\!\!\!\!\!\!\!\!\!\!\! -\frac{1}{2}\log\left( 1-\frac{\theta_1^2}{\sigma_S^2\sigma_X^4} \frac{\theta_2^2}{\sigma_{\hat{X}}^2} \right)}{\geq h(S) - C.}{\label{DCR_C}} 
\end{mini!}

Since both the rate constraint~(\ref{DCR_I}) and the classification constraint~(\ref{DCR_C}) are independent of the mean \(\mu_{\hat{X}}\), and $\theta_2$ only depends on the variance of $X$ and $\hat{X}$ when $X$ and $\hat{X}$ are Gaussian distributions, the objective function is minimized when the means match, i.e., $(\mu_X - \mu_{\hat{X}})^2 + \sigma_X^2 + \sigma_{\hat{X}}^2 - 2\theta_2 \geq \sigma_X^2 + \sigma_{\hat{X}}^2 - 2\theta_2$, we let \(\mu_X = \mu_{\hat{X}}\) in the subsequent derivations.

To ensure that the mutual information expression in~\eqref{DCR_I} is well-defined, it is necessary that \( 1 - \frac{\theta_2^2}{\sigma_X^2 \sigma_{\hat{X}}^2} > 0 \), i.e., \( \frac{\theta_2^2}{\sigma_{\hat{X}}^2} < \sigma_X^2 \). Under this condition, the mutual information between \( S \) and \( \hat{X} \) is upper bounded as
\begin{align*}
\begin{split}
I(S; \hat{X}) &= -\frac{1}{2}\log\left(1 - \frac{\theta_1^2}{\sigma_S^2 \sigma_X^4} \times \frac{\theta_2^2}{\sigma_{\hat{X}}^2} \right),\\ 
&\leq -\frac{1}{2} \log\left(1 - \frac{\theta_1^2}{\sigma_S^2 \sigma_X^2} \right),   
\end{split}
\end{align*}
which implies that constraint~\eqref{DCR_C} becomes infeasible if \( C < \frac{1}{2} \log\left(1 - \frac{\theta_1^2}{\sigma_S^2 \sigma_X^2}\right) + h(S) \). Therefore, to guarantee feasibility, we assume throughout that
\begin{equation*}
C \geq \frac{1}{2} \log\left(1 - \frac{\theta_1^2}{\sigma_S^2 \sigma_X^2} \right) + h(S).
\end{equation*}

Since the function \( D(C, R) \) is convex and monotonically non-increasing in both \( C \) and \( R \), the optimization problem~\eqref{DCR} can be effectively solved using the Karush-Kuhn-Tucker (KKT) conditions. Our approach systematically explores all possible combinations of active and inactive rate and classification constraints to characterize the optimal solution.

\myheading{Case 1.} Constraint~(\ref{DCR_I}) is active and constraint~(\ref{DCR_C}) is inactive.

Recall the classical Shannon rate-distortion function for a Gaussian source \( X \sim \mathcal{N}(\mu_X, \sigma_X^2) \)~\cite{cover1999elements}:
\begin{equation*}\label{eqn:rd_gaussian}
    R(D)={\begin{cases}{\frac {1}{2}}\log (\frac{\sigma _{X}^{2}}{D}),& 0\leq D\leq \sigma _{X}^{2}\\
    0,& D>\sigma _{X}^{2}.\end{cases}}
\end{equation*}
And,
\begin{equation*}\label{eqn:dr_gaussian}
    D(R)=\sigma_X^2 e^{-2R},
\end{equation*}
with the optimal solution attained by some \( p_{\hat{X}|X} \) where \( \hat{X} \sim \mathcal{N}(\mu_X, \sigma_X^2 - \sigma_X^2 e^{-2R}) \). In this case, we have \( D(C, R) = \sigma_X^2 e^{-2R} \), achieved by choosing \( \sigma_{\hat{X}}^2 = \sigma_X^2 - \sigma_X^2 e^{-2R} \) when the rate constraint~\eqref{DCR_I} is active. It implies that: 
\begin{equation*}
    -\frac{1}{2}\log\left(1 - \frac{\theta_2^2}{\sigma_X^2 \sigma_{\hat{X}}^2}\right) = R \Rightarrow \theta_2 = \sigma_X^2 - \sigma_X^2 e^{-2R}.
\end{equation*}

The constraint~(\ref{DCR_C}) is not active if
\begin{align*}
    -\frac{1}{2}\log\left( 1-\frac{\theta_1^2}{\sigma_S^2\sigma_X^4} \frac{\theta_2^2}{\sigma_{\hat{X}}^2} \right) > h(S) - C.
\end{align*}

Hence,
\begin{align*}
C  > \frac{1}{2}\log\left( 1-\frac{\theta_1^2(\sigma_X^2 - \sigma_X^2 e^{-2R})}{\sigma_S^2\sigma_X^4} \right) + h(S). 
\end{align*}

Therefore, $D(C,R) = \sigma_X^2 e^{-2R}$ if $C  > \frac{1}{2}\log\left( 1-\frac{\theta_1^2(\sigma_X^2 - \sigma_X^2 e^{-2R})}{\sigma_S^2\sigma_X^4} \right) + h(S)$.

\myheading{Case 2.}  Constraint~(\ref{DCR_I}) is inactive and constraint~(\ref{DCR_C}) is active.

The classification constraint~\eqref{DCR_C} is active, infer that:
\begin{equation*}
    -\frac{1}{2}\log \left(1-\frac{\theta_1^2}{\sigma_S^2\sigma_X^4} \frac{\theta_2^2}{\sigma_{\hat{X}}^2} \right) = h(S) - C, 
\end{equation*}
\begin{equation}
   \Rightarrow \frac{\theta_2^2}{\sigma_{\hat{X}}^2} = \frac{\sigma_S^2 \sigma_X^4}{\theta_1^2} \left(1 - e^{-2h(S) + 2C}\right). 
    \label{incorporatingC1}   
\end{equation}

We choose $\sigma_{\hat{X}}^2 = \theta_2 = \frac{\sigma_S^2 \sigma_X^4}{\theta_1^2} \left(1 - e^{-2h(S) + 2C}\right)$ and substitute into the distortion expression:
\begin{align*}
    \mathbb{E}[(X - \hat{X}_G)^2] 
    &= \sigma_X^2 + \sigma_{\hat{X}}^2 - 2 \theta_2, \\
    &= \sigma_X^2 - \frac{\sigma_S^2 \sigma_X^4}{\theta_1^2} \left(1 - e^{-2h(S) + 2C}\right).
\end{align*}

Substitute (\ref{incorporatingC1}) into the rate expression (\ref{rateExpression}), we get: 
\begin{align*}
    I(X; \hat{X}) = -\frac{1}{2} \log \left(1 - \frac{\sigma_S^2 \sigma_X^2}{\theta_1^2} \left(1 - e^{-2h(S) + 2C}\right)\right).
\end{align*}

So, the rate constraint~(\ref{DCR_I}) is inactive when 
\begin{equation*}
    -\frac{1}{2} \log \left(1 - \frac{\sigma_S^2 \sigma_X^2}{\theta_1^2} \left(1 - e^{-2h(S) + 2C}\right)\right) < R,
\end{equation*}
\begin{align*}
\Rightarrow C  < \frac{1}{2}\log\left( 1-\frac{\theta_1^2(\sigma_X^2 - \sigma_X^2 e^{-2R})}{\sigma_S^2\sigma_X^4} \right) + h(S).
\end{align*}

Therefore, $D(C,R) = \sigma_X^2 - \frac{\sigma_S^2 \sigma_X^4}{\theta_1^2} \left(1 - e^{-2h(S) + 2C}\right)$ if $C  < \frac{1}{2}\log\left( 1-\frac{\theta_1^2(\sigma_X^2 - \sigma_X^2 e^{-2R})}{\sigma_S^2\sigma_X^4} \right) + h(S)$.

\textbf{Case 3:} Both the rate constraint~(\ref{DCR_I}) and the classification constraint~\eqref{DCR_C} are active.

From case 2, we know that the classification constraint~\eqref{DCR_C} is active if 
\begin{equation*}
    \sigma_{\hat{X}}^2 = \theta_2 = \frac{\sigma_S^2 \sigma_X^4}{\theta_1^2} \left(1 - e^{-2h(S) + 2C}\right), 
\end{equation*}
and, 
\begin{align*}
    \mathbb{E}[(X - \hat{X})^2] = \sigma_X^2 - \frac{\sigma_S^2 \sigma_X^4}{\theta_1^2} \left(1 - e^{-2h(S) + 2C}\right).
\end{align*}

The rate constraint~(\ref{DCR_I}) is active if 
\begin{align*}
    I(X; \hat{X}) = -\frac{1}{2} \log \left(1 - \frac{\sigma_S^2 \sigma_X^2}{\theta_1^2} \left(1 - e^{-2h(S) + 2C}\right)\right) = R.
\end{align*}
\begin{align*}
\Rightarrow C  = \frac{1}{2}\log\left( 1-\frac{\theta_1^2(\sigma_X^2 - \sigma_X^2 e^{-2R})}{\sigma_S^2\sigma_X^4} \right) + h(S).
\end{align*}

Therefore, $D(C,R) = \sigma_X^2 - \frac{\sigma_S^2 \sigma_X^4}{\theta_1^2} \left(1 - e^{-2h(S) + 2C}\right)$ if $C  = \frac{1}{2}\log\left( 1-\frac{\theta_1^2(\sigma_X^2 - \sigma_X^2 e^{-2R})}{\sigma_S^2\sigma_X^4} \right) + h(S)$.

\myheading{Case 4.}  Both constraint \eqref{DCR_I} and constraint \eqref{DCR_C} are inactive.

When \( C > h(S) \), implying that the classification constraint~\eqref{DCR_C} is inactive, and the rate \( R \) is sufficiently large such that \( R > h(X) \), meaning the rate constraint~\eqref{DCR_I} is also inactive, the minimum achievable distortion \( D(C, R) \) reaches its theoretical lower bound, i.e., \( D(C, R) = 0 \). This is achieved by setting \( \hat{X} = X \), which leads to zero reconstruction error, i.e., \( \mathbb{E}[(X - \hat{X})^2] = 0 \). Furthermore, all constraints are satisfied since \( I(X; \hat{X}) = h(X) < R \), and \( h(S | \hat{X}) = h(S | X) \leq h(S) < C \).

Therefore, $D(C,R) = 0$ if $ C > h(S) \text{ and } R > h(X)$. 

In summary, combining the four cases, the closed-form expression for the information distortion-classification-rate function \( D(C, R) \) under MSE distortion is given by~(\ref{Theorem2}).

\subsection{Proof of Theorem \ref{Theorem_gaussian_universality}}\label{Appendix_Proof_GS_Universality}
\begin{figure*}[htbp]
\begin{align}
\label{DC_LowerBoundary}
\begin{split}
&D=\sigma_X^2 - \frac{\sigma_S^2 \sigma_X^4}{\theta_1^2} \left(1 - e^{-2h(S) + 2C}\right),\\
&C\in\left[ \frac{1}{2} \log\left(1 - \frac{\theta_1^2}{\sigma_S^2 \sigma_X^2}\right) + h(S), \frac{1}{2}\log(1-\frac{\theta_1^2 (\sigma_X^2 - \sigma_X^2 e^{-2 R})}{\sigma_S^2\sigma_X^4}) +h(S)  \right).
\end{split}
\end{align}
\hrulefill
\vspace*{-4pt}
\end{figure*}
The proof method follows the approach presented in~\cite{UniversalRDPs}. Let \( R = \sup_{(D,C) \in \Theta} R(D,C) \). By definition, \( \Theta \subseteq \Omega(R) \), where \( \Omega(R) \) denotes the set of all achievable distortion-classification pairs at rate \( R \). The lower boundary of this region, the optimal tradeoff curve, is characterized by Equation~\eqref{DC_LowerBoundary}.

Every point in \( \Omega(R) \) is component-wise dominated by a point on this boundary. Consider a representation \( Z \) that is jointly Gaussian with \( X \), such that \( I(X; Z) = R \). This implies the squared correlation coefficient between \( X \) and \( Z \) satisfies \( \rho_{XZ}^2 = 1 - 2^{-2R} \) \cite{UniversalRDPs}, where 
\begin{equation*}
    \rho_{XZ}= \frac{\text{Cov}(X,Z)}{\sigma_X\sigma_Z} = \frac{\mathbb{E}[(X-\mu_X)(Z-\mu_Z)]}{\sigma_X\sigma_Z}.
\end{equation*}

For any point \( (D, C) \) on the boundary, define the corresponding reconstruction as:
\begin{equation*}
    \hat{X}_{D,C} = \mbox{sign}(\rho_{XZ})\gamma (Z-\mu_Z)+\mu_X,
\end{equation*}
where \( \gamma \) denotes a scaling coefficient and
\begin{align*}
    \mbox{sign}(\rho_{XZ}) = \begin{cases}
1, &\text{for } \rho_{XZ}\geq 0,\\
-1, &\text{for } \rho_{XZ} < 0.
\end{cases}
\end{align*}

With this construction, we have:
\begin{equation*}
\begin{split}
    \mu_{\hat{X}_{D,C}} &= \mathbb{E}[\mbox{sign}(\rho_{XZ})\gamma (Z-\mu_Z)+\mu_X], \\
    &= \mbox{sign}(\rho_{XZ}) \gamma (\mathbb{E}[Z] - \mathbb{E}[Z]) + \mu_X,\\
    &= \mu_{X}.
\end{split}
\end{equation*}

And,
\begin{equation*}
    \sigma_{\hat{X}_{D,C}} = \gamma^2 \sigma_{Z},
\end{equation*}
\begin{equation*}
    \text{Cov}(X, \hat{X}_{D,C})  = \gamma \text{Cov}(X, Z).
\end{equation*}

We now choose:
\begin{equation*}
\sigma_{\hat{X}_{D,C}} = \theta_2 = \frac{\sigma_S^2 \sigma_X^4 (1 - e^{-2 h(S) + 2C})}{\theta_1^2}.
\end{equation*}

Solving for \( \gamma \), we obtain: 
\begin{equation*}
    \gamma = \frac{\sigma_S \sigma_X^2 \sqrt{1-e^{-2 h(S) + 2C}}}{\theta_1 \sigma_Z}.
\end{equation*}

Hence,
\begin{equation*}
    \hat{X}_{D,C} = \mbox{sign}(\rho_{XZ}) \frac{\sigma_S \sigma_X^2 \sqrt{1-e^{-2 h(S) + 2C}}}{\theta_1 \sigma_Z} (Z-\mu_Z)+\mu_X.
\end{equation*}

We now verify that this reconstruction satisfies the distortion and classification constraints.

\myheading{Distortion constraint.}
\begin{align*}
\mathbb{E}[\| X-\hat{X}_{D,C} \|^2] &=\sigma^2_X+\sigma^2_{\hat{X}_{D,C}}-2\theta_2,\\
&= \sigma_X^2 - \frac{\sigma_S^2 \sigma_X^4}{\theta_1^2} \left(1 - e^{-2h(S) + 2C}\right),\\
&= D.
\end{align*} 

\myheading{Classification constraint.}
\begin{align*}
h(S|\hat{X}_{D,C}) &= h(S) - I(S|\hat{X}_{D,C}), \\
&= h(S) + \frac{1}{2} \log\left( 1 - \frac{\theta_1^2}{\sigma_S^2 \sigma_X^4 } \frac{\theta_2^2}{\sigma_{\hat{X}_{D,C}}^2} \right), 
\end{align*}
where  
\begin{equation*}
    \frac{\theta_2^2}{\sigma_{\hat{X}_{D,C}}^2} = \frac{\sigma_S^2 \sigma_X^4 (1 - e^{-2 h(S) + 2C})}{\theta_1^2}.
\end{equation*}

Substituting in, we get:
\begin{align*}
h(S|\hat{X}_{D,C}) &= h(S) \\ 
&+ \frac{1}{2} \log\left( 1 - \frac{\theta_1^2}{\sigma_S^2 \sigma_X^4 } \frac{\sigma_S^2 \sigma_X^4 (1 - e^{-2 h(S) + 2C})}{\theta_1^2} \right), \\
&= h(S) + \frac{1}{2} \log \left( e^{-2h(S) + 2C} \right),\\
&= C. 
\end{align*}

This confirms that for any given \( \text{Cov}(X, Z) \), one can always choose a scalar \( \gamma \) to generate \( \hat{X}_{D,C} \) satisfying both distortion and classification constraints. That is, every point \( (D, C) \in \Theta \) can be realized by applying an appropriate decoder to a fixed Gaussian representation \( Z \) of \( X \) such that \( I(X; Z) = \sup_{(D, C) \in \Theta} R(D, C) \). Therefore, \( \Omega(p_{Z|X}) = \Omega(R) \), which implies the rate penalty is zero: $A(\Theta) = I(X; Z) - R = 0$.

\subsection{Proof of Theorem \ref{Theorem_general_universality}} \label{Appendix_Proof_General_Universality}
The proof approach is inspired by the main techniques introduced in~\cite{UniversalRDPs}. For any \( (D, C) \in \Omega(p_{Z|X}) \), there exists a reconstruction variable \( \hat{X}_{D,C} \) jointly distributed with \( (X, Z) \), such that the Markov chain \( X \leftrightarrow Z \leftrightarrow \hat{X}_{D,C} \) holds, the distortion satisfies \( \mathbb{E}[\Delta(X, \hat{X}_{D,C})] \leq D \), and the classification uncertainty is bounded by \( H(S | \hat{X}_{D,C}) \leq C \). 

Since $\tilde{X} = \mathbb{E}[X|Z]$, we have:
\begin{align*}
D\geq\mathbb{E}[\|X-\hat{X}_{D,C}\|^2]  =\mathbb{E}[\|X-\tilde{X}\|^2]+\mathbb{E}[\|\tilde{X}-\hat{X}_{D,C}\|^2].
\end{align*}

Now consider the Wasserstein-2 distance between the marginals \( p_{\tilde{X}} \) and \( p_{\hat{X}_{D,C}} \), which is defined as:
\[
W_2^2(p_{\tilde{X}}, p_{\hat{X}_{D,C}}) = \inf_{p_{\tilde{X}', \hat{X}'}} \mathbb{E}[\|\tilde{X}' - \hat{X}'\|^2],
\]
where \( \tilde{X}' \sim p_{\tilde{X}} \) and \( \hat{X}' \sim p_{\hat{X}_{D,C}} \). Since \( (\tilde{X}, \hat{X}_{D,C}) \) is one feasible coupling of these marginals,
\[
\mathbb{E}[\|\tilde{X} - \hat{X}_{D,C}\|^2] \geq W_2^2(p_{\tilde{X}}, p_{\hat{X}_{D,C}}).
\]

Therefore,
\[
D \geq \mathbb{E}[\|X - \tilde{X}\|^2] + W_2^2(p_{\tilde{X}}, p_{\hat{X}_{D,C}}).
\]

Since \( H(S | \hat{X}_{D,C}) \leq C \), the marginal distribution \( p_{\hat{X}_{D,C}} \) belongs to the constraint set \( \{ p_{\hat{X}} : H(S | \hat{X}) \leq C \} \). We have, $p_{\hat{X}_{D,C}}$ is one feasible distribution of the set $\left\{ p_{\hat{X}}:\;
\begin{aligned}
    &\inf_{p_{\hat{X}}}   W^2_2(p_{\tilde{X}},p_{\hat{X}}) \\
    &\text{s.t. } H(S|\hat{X}) \leq C
\end{aligned}
\right\}$, then
\[
W_2^2(p_{\tilde{X}}, p_{\hat{X}_{D,C}}) \geq 
\left\{
\begin{aligned}
&\inf_{p_{\hat{X}}} W_2^2(p_{\tilde{X}}, p_{\hat{X}}) \\
&\text{s.t. } H(S | \hat{X}) \leq C.
\end{aligned}
\right.
\]

This leads to the outer bound:
\begin{equation*}
\Omega(p_{Z|X}) \! \subseteq \! \left\{ \! (D,C) \! : \! D  \geq  \mathbb{E}[{\|X-\tilde{X}\|^2}] \!+\! \begin{aligned}
    &\inf_{p_{\hat{X}}}  W^2_2(p_{\tilde{X}},p_{\hat{X}}) \\
    &\text{s.t. } H(S|\hat{X}) \leq C.
\end{aligned} \! \right\}
\end{equation*}

Now, to show the approximate tightness of this bound, let \( (D', C') \) be any point in the above region. For any \( \epsilon > 0 \), there exists a distribution \( p_{\hat{X}'} \) such that:
\[
H(S |\hat{X}') \leq C', \quad D' + \epsilon \geq \mathbb{E}[\|X - \tilde{X}\|^2] + W_2^2(p_{\tilde{X}}, p_{\hat{X}'}).
\]

By the Markov condition, we can construct a random variable \( \hat{X}' \) such that \( X \leftrightarrow Z \leftrightarrow \hat{X}' \), and
\[
\mathbb{E}[\|\tilde{X} - \hat{X}'\|^2] \leq W_2^2(p_{\tilde{X}}, p_{\hat{X}'}) + \epsilon.
\]
Thus,
\[
\mathbb{E}[\|X - \hat{X}'\|^2] = \mathbb{E}[\|X - \tilde{X}\|^2] + \mathbb{E}[\|\tilde{X} - \hat{X}'\|^2] \leq D' + 2\epsilon.
\]

It follows that:
\begin{equation*}
\begin{split}
    \Omega(p_{Z|X}) \! &\subseteq \! \left\{ (D,C) : D \geq \mathbb{E}{\|X-\tilde{X}\|^2} \! + \! \begin{aligned}
    &\inf_{p_{\hat{X}}}   W^2_2(p_{\tilde{X}},p_{\hat{X}}) \\
    &\text{s.t. } H(S|\hat{X}) \leq C
\end{aligned} \right\} \\
    \! & \subseteq \! \mbox{cl}(\Omega(p_{Z|X})).
\end{split}
\end{equation*}

Now consider the characterization of conditional entropy:
\begin{align*}
    H(S|\hat{X})\! &=\!\! \sum_{s}\sum_{\hat{x}} p_{S,\hat{X}}\log\frac{1}{p_{S|\hat{X}}},\\
    &=\!\! \sum_{s}\sum_{\hat{x}} p_{\hat{X}} p_{S|\hat{X}}\log\frac{1}{p_{S|\hat{X}}}.
\end{align*}

By choosing \( \hat{X} = \tilde{X} \), it follows that \( p_{\hat{X}} = p_{\tilde{X}} \), which yields:
\begin{equation*}
     H(S|\hat{X}) = H(S|\tilde{X}) = \sum_{s}\sum_{\tilde{x}} p_{\tilde{X}} p_{S|\tilde{X}} \log\frac{1}{p_{S|\tilde{X}}},
\end{equation*}
and define:
\begin{equation*}
(D^{(a)},C^{(a)})=\left( \mathbb{E}[\|X-\tilde{X}\|^2], \sum_{s}\sum_{\tilde{x}} p_{\tilde{X}} p_{S|\tilde{X}} \log\frac{1}{p_{S|\tilde{X}}} \right).   
\end{equation*}

Next, define $\hat{X}^{C_{\text{min}}} \sim p_{\hat{X}^{C_{\text{min}}}}$:
\begin{mini!}|s|[2] 
{p_{\hat{X}}} 
{ H(S|\hat{X})} 
{\label{C_min_Appendix}} 
{ p_{\hat{X}^{C_{\text{min}}}} = \arg} 
\addConstraint{\mathbb{E}[\| X-\hat{X} \|^2]}{\leq D.} 
\end{mini!}
And,
\begin{equation*}
    C_{\min} = \sum_{s}\sum_{\hat{x}^{C_{\text{min}}}} p_{\hat{X}^{C_{\text{min}}}} p_{S|\hat{X}^{C_{\text{min}}}} \log\frac{1}{p_{S|\hat{X}^{C_{\text{min}}}}}.
\end{equation*}

Let $\hat{X} = \hat{X}^{C_{\text{min}}}$ implies \( p_{\hat{X}} = p_{\hat{X}^{C_{\text{min}}}} \), and define:
\[
(D^{(b)}, C^{(b)}) = \left( \mathbb{E}[\|X - \tilde{X}\|^2] + W_2^2(p_{\tilde{X}}, p_{\hat{X}^{C_{\text{min}}}}), C_{\min} \right).
\]

Thus, by selecting \( p_{\hat{X}} = p_{\tilde{X}} \) and \( p_{\hat{X}} =  p_{\hat{X}^{C_{\text{min}}}} \), we confirm that both \( (D^{(a)}, C^{(a)}) \) and \( (D^{(b)}, C^{(b)}) \) lie in this region: 
\[
\left\{ (D,C) : D \geq \mathbb{E}{\|X-\tilde{X}\|^2} + \begin{aligned}
    &\inf_{p_{\hat{X}}}  W^2_2(p_{\tilde{X}},p_{\hat{X}}) \\
    &\text{s.t. } H(S|\hat{X}) \leq C
\end{aligned} \right\}.
\]

\subsection{Proof of Theorem \ref{Theorem_Quantitative_Results}}\label{Appendix_Proof_Quantitative_Results}
The proof idea follows the result in~\cite{UniversalRDPs}. We begin by noting that \( C_3 = C_{\min} \), and from the hypothesis \( R(D_3, C_3) = R(D_1, \infty) \). Next, observe that:
\begin{align*}
D_3&=\mathbb{E}[\|X-\hat{X}_{D_3,C_3}\|^2], \\
&= \sigma^2_X + \sigma_{\hat{X}_{D_3, C_3}}^2 -2 \text{Cov}(X, \hat{X}_{D_3,C_3}), \\
&=\sigma^2_X + \sigma_{\hat{X}_{D_3, C_3}}^2 -2\mathbb{E}[(X-\mu_X)^T(\hat{X}_{D_3,C_3}-\mu_{\hat{X}_{D_3, C_3}})].
\end{align*} 

Thus,
\begin{equation*}
\label{Eq_D3}
\mathbb{E}[(X-\mu_X)^T(\hat{X}_{D_3,C_3}-\mu_{\hat{X}_{D_3, C_3}}) = \frac{\sigma^2_X + \sigma_{\hat{X}_{D_3, C_3}}^2 - D_3}{2}.
\end{equation*}

We now utilize the inequality \( I(X; \mathbb{E}[X | \hat{X}_{D_3,C_3}]) \leq I(X; \hat{X}_{D_3,C_3}) = R(D_1, \infty) \), which implies: $\mathbb{E}[\|X-\mathbb{E}[X|\hat{X}_{D_3,C_3}]\|^2]\geq D_1$. 

Observe that $\hat{X}_{D_3,C_3} - \mu_{\hat{X}_{D_3, C_3}}$ has a certain correlation with $X - \mu_X$, so we can use a linear predictor idea to get a simpler upper bound. By the orthogonality principle:
\begin{equation*}
\begin{split}
&\mathbb{E}[\|X  -  \mathbb{E}[X | \hat{X}_{D_3,C_3}]\|^2]  \\ 
&\leq \mathbb{E}[\|X - \mu_X - c (\hat{X}_{D_3,C_3} - \mu_{\hat{X}_{D_3, C_3}})\|^2].
\end{split}
\end{equation*}
The inequality says the best predictor $\mathbb{E}[X | \hat{X}_{D_3,C_3}]$ is no worse (in MSE) than any fixed linear predictor of the form $\mu_X + c (\hat{X}_{D_3,C_3} - \mu_{\hat{X}_{D_3, C_3}})$, where
\begin{equation*}
\begin{split}
c \! &= \!\! \frac{\text{Cov}(X-\mu_X, \hat{X}_{D_3,C_3} - \mu_{\hat{X}_{D_3, C_3}})}{\text{Var}({\hat{X}_{D_3,C_3} - \mu_{\hat{X}_{D_3, C_3}}})},\\
& \!\! = \!\! \frac{\mathbb{E}[(X-\mu_X)^T(\hat{X}_{D_3,C_3}-\mu_{\hat{X}_{D_3, C_3}})]}{\sigma_{\hat{X}_{D_3, C_3}}^2} \\
& \!\! = \!\! \frac{\sigma^2_X + \sigma_{\hat{X}_{D_3, C_3}}^2 - D_3}{2\sigma_{\hat{X}_{D_3, C_3}}^2}.
\end{split}
\end{equation*}

Therefore, the MSE of the optimal conditional expectation is at most the MSE of this linear estimator.  
\begin{align*}
D_1&\leq\mathbb{E}[\|X-\mathbb{E}[X|\hat{X}_{D_3,C_3}]\|^2],\\
&\leq\mathbb{E}[\|X-\mu_X-c(\hat{X}_{D_3,C_3}-\mu_{\hat{X}_{D_3, C_3}})\|^2],\\
&= \mathbb{E}[\|X - \mu_X\|^2] + c^2 \mathbb{E}[\|\hat{X}_{D_3,P_3} - \mu_{\hat{X}_{D_3, C_3}}\|^2] \\
&- 2c \mathbb{E}[(X - \mu_X)^T (\hat{X}_{D_3,P_3} - \mu_{\hat{X}_{D_3, C_3}})],\\
&=\sigma^2_X + c^2 \sigma_{\hat{X}_{D_3, C_3}}^2 \!\!\!\! - 2c\mathbb{E}[(X-\mu_X)^T(\hat{X}_{D_3,C_3} \! - \! \mu_{\hat{X}_{D_3, C_3}})],\\
&=\sigma^2_X + \left( \frac{\sigma^2_X + \sigma_{\hat{X}_{D_3, C_3}}^2 - D_3}{2\sigma_{\hat{X}_{D_3, C_3}}^2} \right)^2 \sigma_{\hat{X}_{D_3, C_3}}^2 \\
&-2 \left( \frac{\sigma^2_X + \sigma_{\hat{X}_{D_3, C_3}}^2 - D_3}{2\sigma_{\hat{X}_{D_3, C_3}}^2} \right) \left( \frac{\sigma^2_X + \sigma_{\hat{X}_{D_3, C_3}}^2 - D_3}{2} \right),\\
&= \sigma^2_X - \frac{(\sigma^2_X + \sigma_{\hat{X}_{D_3, C_3}}^2 - D_3)^2}{4 \sigma_{\hat{X}_{D_3, C_3}}^2}.
\end{align*}

Rearranging terms yields:
\begin{equation*}
(\sigma^2_X + \sigma_{\hat{X}_{D_3, C_3}}^2 - D_3)^2 \leq 4 \sigma_{\hat{X}_{D_3, C_3}}^2 (\sigma^2_X - D_1).   
\end{equation*}

Under the assumptions \( \sigma_X^2 + \sigma_{\hat{X}_{D_3,C_3}}^2 - D_3 \geq 0 \) and \( \sigma_X^2 - D_1 \geq 0 \), it follows that:
\begin{equation*}
\begin{split}
\sigma^2_X + \sigma_{\hat{X}_{D_3, C_3}}^2 - D_3 \leq 2 \sigma_{\hat{X}_{D_3, C_3}} \sqrt{\sigma_X^2 - D_1},\\
D_3 \geq \sigma^2_X + \sigma_{\hat{X}_{D_3, C_3}}^2 - 2 \sigma_{\hat{X}_{D_3, C_3}} \sqrt{\sigma_X^2 - D_1}.
\end{split}
\end{equation*}

Based on~\cite{blau2018perception}, we can show that:
\begin{equation*}
D^{(b)} \leq D_3 \leq 2D_1.
\end{equation*}
Hence, 
\begin{align*}
&D_3 - D^{(b)}\geq \sigma^2_X + \sigma_{\hat{X}_{D_3, C_3}}^2 - 2 \sigma_{\hat{X}_{D_3, C_3}} \sqrt{\sigma_X^2 - D_1} - 2D_1,\\
&\frac{D_3}{D^{(b)}}\geq \frac{\sigma^2_X + \sigma_{\hat{X}_{D_3, C_3}}^2 - 2 \sigma_{\hat{X}_{D_3, C_3}} \sqrt{\sigma_X^2 - D_1}}{2D_1}.
\end{align*}

\underline{For $D_1 = 0$:}
\begin{align*}
D_3 - D^{(b)} &\geq \sigma^2_X + \sigma_{\hat{X}_{D_3, C_3}}^2 - 2 \sigma_{\hat{X}_{D_3, C_3}} \sigma_X,
\end{align*}

Moreover, in the case of $W_2^2(p_X, p_{\hat{X}_{D_3, C_3}}) = 0$ (i.e., $\sigma_X^2 = \sigma_{\hat{X}_{D_3, C_3}}^2$), we have:
\begin{equation*}
D_3 - D^{(b)}\stackrel{D_1\approx0 \hspace{1mm} \text{and} \hspace{1mm} \sigma_X^2 = \sigma_{\hat{X}_{D_3, C_3}}^2}{\approx}0, 
\end{equation*}

\underline{For $D_1 = \sigma_X^2$:}
\begin{align*}
&D_3 - D^{(b)} \geq \sigma^2_X + \sigma_{\hat{X}_{D_3, C_3}}^2 - 2\sigma^2_X,\\
&\frac{D_3}{D^{(b)}} \geq \frac{\sigma^2_X + \sigma_{\hat{X}_{D_3, C_3}}^2}{2 \sigma^2_X}.
\end{align*}

Again, in the case of $W_2^2(p_X, p_{\hat{X}_{D_3, C_3}}) = 0$, then:
\begin{align*}
&D_3 - D^{(b)}\stackrel{D_1\approx\sigma_X^2 \hspace{1mm} \text{and} \hspace{1mm} \sigma_X^2 = \sigma_{\hat{X}_{D_3, C_3}}^2}{\approx}0, \\
&\frac{D_3}{D^{(b)}}\stackrel{D_1\approx\sigma_X^2 \hspace{1mm} \text{and} \hspace{1mm} \sigma_X^2 = \sigma_{\hat{X}_{D_3, C_3}}^2}{\approx} 1.
\end{align*}

A similar method can be utilized to bound the upper-left corner of the curve, i.e., the gap between $(D_1,C_1)$ and the upper-left extreme point $(\tilde{D}^{(a)},\tilde{C}^{(a)})$ of the blue curve in Fig.~\ref{Fig:Universal_GernalSource_CDR}. Let \( \tilde{D}^{(a)} = \mathbb{E}[\|X - \mathbb{E}[X \mid \hat{X}_{D_3,C_3}]\|^2] \). Then:
\[
\tilde{D}^{(a)} \leq \sigma_X^2 - \frac{(\sigma_X^2 + \sigma_{\hat{X}_{D_3,C_3}}^2 - D_3)^2}{4 \sigma_{\hat{X}_{D_3,C_3}}^2},
\]
which, together with \( D_1 \geq \frac{1}{2} D_3 \), yields:
\begin{align*}
&\tilde{D}^{(a)}-D_1 \leq \sigma^2_X - \frac{(\sigma^2_X + \sigma_{\hat{X}_{D_3, C_3}}^2 - D_3)^2}{4 \sigma_{\hat{X}_{D_3, C_3}}^2} - \frac{D_3}{2},\\
&\frac{\tilde{D}^{(a)}}{D_1}\leq \frac{\sigma^2_X - \frac{(\sigma^2_X + \sigma_{\hat{X}_{D_3, C_3}}^2 - D_3)^2}{4 \sigma_{\hat{X}_{D_3, C_3}}^2}}{D_3/2}.
\end{align*}

\underline{For $D_3 = 0$:}
\begin{align*}
\tilde{D}^{(a)} - D_1 &\leq \sigma^2_X - \frac{(\sigma^2_X + \sigma_{\hat{X}_{D_3, C_3}}^2)^2}{4 \sigma_{\hat{X}_{D_3, C_3}}^2},
\end{align*}

If of $W_2^2(p_X, p_{\hat{X}_{D_3, C_3}}) = 0$, then:
\begin{equation*}
\tilde{D}^{(a)}-D_1 \stackrel{D_3\approx0 \hspace{1mm} \text{and} \hspace{1mm} \sigma_X^2 = \sigma_{\hat{X}_{D_3, C_3}}^2}{\approx}0, 
\end{equation*}

\underline{For $D_3 = 2\sigma_X^2$:}
\begin{align*}
&\tilde{D}^{(a)}-D_1 \leq \sigma^2_X - \frac{( \sigma_{\hat{X}_{D_3, C_3}}^2 - \sigma^2_X)^2}{4 \sigma_{\hat{X}_{D_3, C_3}}^2} - \sigma^2_X,\\
&\frac{\tilde{D}^{(a)}}{D_1} \leq \frac{\sigma^2_X - \frac{( \sigma_{\hat{X}_{D_3, C_3}}^2 - \sigma^2_X)^2}{4 \sigma_{\hat{X}_{D_3, C_3}}^2}}{\sigma_X^2}.
\end{align*}

Similarly, in the case of $W_2^2(p_X, p_{\hat{X}_{D_3, C_3}}) = 0$, we have:
\begin{align*}
&\tilde{D}^{(a)}-D_1 \stackrel{D_3\approx 2\sigma_X^2 \hspace{1mm} \text{and} \hspace{1mm} \sigma_X^2 = \sigma_{\hat{X}_{D_3, C_3}}^2}{\approx}0, \\
&\frac{\tilde{D}^{(a)}}{D_1} \stackrel{D_3\approx 2\sigma_X^2 \hspace{1mm} \text{and} \hspace{1mm} \sigma_X^2 = \sigma_{\hat{X}_{D_3, C_3}}^2}{\approx} 1.
\end{align*}

\end{document}